\documentclass[journal,twocolumn]{IEEEtran}
\usepackage{bbm}
\usepackage{array}
\usepackage{dcolumn}
\usepackage{epsfig}
\usepackage[intlimits]{amsmath}
\usepackage{amsmath, amsfonts, yhmath, bm}
\usepackage{amssymb}
\usepackage{psfrag}
\usepackage{color,soul}
\usepackage[dvipsnames]{xcolor}
\usepackage[normalem]{ulem}
\usepackage{enumerate}
\usepackage{stackengine}
\usepackage[noadjust]{cite}
\usepackage{graphicx}
\usepackage{caption}
\usepackage{subcaption}
\usepackage[font=footnotesize]{subcaption}
\usepackage{multirow}
\usepackage{cite}
\usepackage[font=footnotesize]{caption}
\usepackage{etoolbox}
\usepackage{tcolorbox}
\usepackage{float}
\usepackage{dsfont}
\usepackage{tikz}
\usetikzlibrary{arrows}

\newcommand {\myvec}[1] {{\mbox{\boldmath $#1$}}}
\newcommand {\mymat}[1]  {{\mbox{\boldmath $#1$}}}

\DeclareMathAlphabet      {\mathbfit}{OML}{cmm}{b}{it}

\newcommand {\mLambda} {\mymat{\Lambda}}

\newcommand {\mU} {\mymat{U}}

\newcommand {\I} {\mymat{I}}

\newcommand {\uxi} {\myvec{\xi}}

\newcommand {\ux} {\myvec{x}}

\newcommand {\uo} {\myvec{0}}

\newcommand {\uy} {\myvec{y}}

\newcommand {\Rset} {\mathbb{R}}
\newcommand {\Cset} {\mathbb{C}}

\newcommand {\Eset} {\mathbb{E}}

\newcommand {\herm} {\rm{H}}

\newcommand {\LS} {\mbox{\tiny LS}}
\newcommand {\GSW} {\mbox{\tiny GSW}}

\makeatletter
\newsavebox\myboxA
\newsavebox\myboxB
\newlength\mylenA

\newcommand*\mybar[2][0.75]{%
	\sbox{\myboxA}{$\m@th#2$}%
	\setbox\myboxB\null
	\ht\myboxB=\ht\myboxA%
	\dp\myboxB=\dp\myboxA%
	\wd\myboxB=#1\wd\myboxA
	\sbox\myboxB{$\m@th\overline{\copy\myboxB}$}
	\setlength\mylenA{\the\wd\myboxA}
	\addtolength\mylenA{-\the\wd\myboxB}%
	\ifdim\wd\myboxB<\wd\myboxA%
	\rlap{\hskip 0.5\mylenA\usebox\myboxB}{\usebox\myboxA}%
	\else
	\hskip -0.5\mylenA\rlap{\usebox\myboxA}{\hskip 0.5\mylenA\usebox\myboxB}%
	\fi}
\makeatother

\makeatletter
\newcommand*{\defeq}{\mathrel{\rlap{%
			\raisebox{0.3ex}{$\m@th\cdot$}}%
		\raisebox{-0.3ex}{$\m@th\cdot$}}%
	=}
\makeatother

\newtheorem{remark}{Remark}

\newtheorem{proposition}{Proposition}
\begin{document}

\title{GSW: Generalized ``Self-Wiener" Denoising}

\author{Amir~Weiss\\
         Faculty of Engineering \\
        Bar-Ilan University \\
        Ramat-Gan, Israel \\
        amir.weiss@biu.ac.il}

\maketitle

\begin{abstract}
We revisit the recently proposed ``self-Wiener" (SW) filtering method for robust deconvolution, and generalize it to the classical denoising problem. The resulting estimator, termed generalized SW (GSW) filtering, retains the nonlinear shrinkage structure of SW but introduces a tunable threshold parameter. This tunability enables GSW to flexibly adapt to varying signal-to-noise ratio (SNR) regimes by balancing noise suppression and signal preservation. We derive closed-form expressions for its mean-square error (MSE) performance in both low- and high-SNR regimes, and demonstrate that GSW closely approximates the oracle MMSE at high SNR while maintaining strong robustness at low SNR. Simulation results validate the analytical findings, showing that GSW consistently achieves favorable denoising performance across a wide range of SNRs. Its analytical tractability, parameter flexibility, and close connection to the optimal Wiener filter structure make it a promising tool for practical applications including compressive sensing, sparse signal recovery, and domain-specific shrinkage in wavelet, Fourier, and potentially learned orthonormal representations.
\end{abstract}


\IEEEpeerreviewmaketitle

\vspace{-0.1cm}
\section{Introduction}\label{sec:intro}
Denoising remains a fundamental task in signal processing, with broad relevance across domains such as communications, imaging, and geophysical sensing (e.g., \cite{buades2005nonlocal,benesty2009noisereduction,loizou2013speechenhancement,meyr1998digital,schmelzbach2015efficient}). In many applications, one ultimately needs accurate estimates of a signal-of-interest (SOI) from noisy measurements due to physical effects of additive noise, either directly in the domain they were measured (e.g., time or space) or after a transform (e.g., Fourier or wavelet~\cite{mallat2009wavelet}) to a more convenient, or in some sense more natural, representation. Classical approaches range from linear Wiener filtering~\cite{wiener1949extrapolation}---optimal under Gaussian statistical models with \emph{known} second-order statistics---to nonlinear shrinkage and thresholding methods designed for sparse or compressible representations~\cite{donoho1994ideal,donoho1995denoising,dabov2007bm3d}.

A recent contribution in this broader space is the ``self-Wiener'' (SW) filter~\cite{weiss2021self}, originally proposed for robust deconvolution of deterministic signals. There, under an approximate independence assumption between discrete Fourier transform (DFT) components, SW applies a nonlinear thershold-type shrinkage operator to each DFT coefficient of the deconvolved signal. Analytically, it was shown that in the high signal-to-noise ratio (SNR) regime, this nonlinear operator closely approximates the ideal (but unrealizable) Wiener filter, while at low SNR it acts as an aggressive noise suppressor.

However, the original SW formulation, which was tailored to deconvolution with a known linear time-invariant system, uses a \emph{fixed} threshold value. This limits its adaptability when applied to pure denoising problems, where one typically has more freedom in choosing the operating point, and where transform-domain sparsity or compressibility can often be exploited explicitly. In particular, in low SNR regimes there is a benefit in allowing the user to bias the estimator toward stronger noise suppression, whereas at high SNR one would like to remain as close as possible to the oracle minimum mean-square error (MMSE) solution, or at least to the oracle linear MMSE solution, obtained via Wiener filtering.

In this work, we revisit the core SW idea from a denoising viewpoint and specialize it to the canonical additive white Gaussian noise (AWGN) model. The resulting estimator, which we term \emph{generalized self-Wiener} (GSW) filtering, retains the scalar, coefficient-wise nonlinear shrinkage structure of SW but introduces an explicit tunable threshold parameter. This tunability allows GSW to interpolate between an aggressively regularized, highly robust regime and a near-oracle MMSE regime, and makes it directly applicable as a drop-in shrinkage rule in generic orthonormal representations, such as Fourier, wavelet, and potentially learned transforms.

Our key contributions in this work are the following:
\begin{itemize}
  \item \emph{GSW filtering for denoising:} We introduce a tunable generalization of the SW filter~\cite{weiss2021self} for the standard AWGN denoising problem. GSW preserves the Wiener-like shrinkage structure of SW but adds flexibility in the threshold parameter that can be matched to the noise level or chosen according to classical prescriptions (e.g., universal or SURE-based thresholds~\cite{donoho1994ideal,luisier2007new}).
  \item \emph{Analytical performance characterization:} We derive closed-form expressions for the mean-square error (MSE) of GSW in both low- and high-SNR regimes and show that, for appropriate threshold choices, GSW closely tracks the oracle MMSE at high SNR while maintaining strong robustness and noise suppression at low SNR.
  \item \emph{Empirical validation and comparison:} We compare the GSW to established benchmarks such as the least-squares (LS) estimator, James--Stein (JS) shrinkage~\cite{draper1979ridge}, and soft-thresholding (ST)~\cite{donoho1995denoising} in a simulation experiment on synthetic data in a sparse vector denoising task, and demonstrate that GSW provides competitive or superior performance across a wide range of noise conditions.
\end{itemize}

The rest of the paper is organized as follows. Section~\ref{sec:problemformulation} formalizes the denoising problem. Section~\ref{sec:gws} introduces the GSW estimator and presents the MSE performance analysis in the low- and high-SNR regimes. Section~\ref{sec:simulationresults} provides numerical simulations, and Section~\ref{sec:conclusion} concludes the paper.

\section{Problem Formulation}\label{sec:problemformulation}

We consider the canonical signal-plus-noise model
\begin{equation}
    \uy = \ux + \sigma \uxi \in \Cset^{N \times 1},
\end{equation}
where $\ux \in \Cset^{N \times 1}$ is an unknown deterministic SOI, $\uxi \sim \mathcal{CN}(\uo,\I_N)$ is an standard, circularly symmetric complex normal vector with independent entries, and $\sigma \in \Rset_{+}$ denotes the noise standard deviation.\footnote{Throughout the paper, all randomness is due to the noise $\uxi$; the SOI $\ux$ is treated as deterministic.} Thus, componentwise,
\begin{equation}\label{eq:scalarmodel}
    y_n = x_n + \sigma \xi_n, \qquad n = 1,\ldots,N,
\end{equation}
with $\xi_n \sim \mathcal{CN}(0,1)$.

This model also covers transform-domain denoising under any unitary transform $\mU \in \Cset^{N \times N}$: setting $\tilde{\uy} \triangleq \mU^{\herm} \uy$ and $\tilde{\ux} \triangleq \mU^{\herm} \ux$ yields
\begin{equation}
    \tilde{\uy} = \tilde{\ux} + \sigma \tilde{\uxi},
\end{equation}
with $\tilde{\uxi} \sim \mathcal{CN}(\uo,\I_N)$ by unitary invariance. Hence, without loss of generality, we may regard $\uy$ as either a time-/space-domain vector or a collection of transform coefficients and focus on estimators that act componentwise on $\{y_n\}_{n=1}^N$.

Our goal is to estimate $\ux$ from its noisy observation $\uy$ by means of an estimator $\widehat{\ux} = \widehat{\ux}(\uy)$, and to evaluate its performance via the MSE,
\begin{align}
    \mathsf{MSE}(\ux,\widehat{\ux})
    &\triangleq \Eset\left[\|\widehat{\ux}(\uy) - \ux\|_2^2\right] \label{eq:defMSEvec} \\
    &= \sum_{n=1}^{N} \Eset\left[|\hat{x}_n(y_n) - x_n|^2\right], \label{eq:defMSEsum}
\end{align}
where the expectation is with respect to the noise $\uxi$, and in the second equality we have emphasized the scalar, componentwise form of the estimator, $\hat{x}_n = \hat{x}_n(y_n)$, that will be the focus of this work.

We also define the per-component SNR,
\begin{equation}
    \mathsf{SNR}(x_n,\sigma) \triangleq \frac{|x_n|^2}{\sigma^2},
\end{equation}
and we will use the terms ``low-SNR'' and ``high-SNR'' regimes to refer to the limits $\mathsf{SNR}(x_n,\sigma) \to 0$ and $\mathsf{SNR}(x_n,\sigma) \to \infty$, respectively.

To keep the exposition simple, we assume homoscedastic noise throughout, i.e., $\uxi \sim \mathcal{CN}(\uo,\I_N)$. The extension to a heteroscedastic noise $\uxi \sim \mathcal{CN}(\uo,\mLambda)$, with a positive-definite diagonal covariance $\mLambda \succ \uo$, is straightforward and does not alter the core analysis, provided the noise variances are known or can be estimated and appropriately normalized per component.

\section{Generalized ``Self-Wiener'' Filtering}\label{sec:gws}
In this section, we introduce GSW filtering, or the GSW estimator, for model \eqref{eq:scalarmodel}. We first define the estimator as a nonlinear shrinkage of the LS solution and then analyze its MSE in the high- and low-SNR regimes.

\subsection{Definition and Intuition}\label{subsec:gws_definition}
It is well known that, under the model $\uy = \ux + \sigma \uxi$, the LS (and
maximum-likelihood) estimator of $\ux$ is given by
\begin{equation}\label{eq:lsest}
    \widehat{\ux}_{\LS} \triangleq \uy.
\end{equation}
Building on the structure of the SW filter~\cite[Eq.~(17)]{weiss2021self}, we
define the generalized SW (GSW) estimator as a nonlinear shrinkage of the LS
estimate, parameterized by a \emph{tunable} threshold $\lambda \in [2,\infty)$. Specifically,
the GSW estimator of the $n$-th entry of $\ux$ is given by
\begin{equation}\label{eq:GSWest}
    \widehat{x}_{n,\GSW} \triangleq y_n \cdot
    \begin{cases}
        \displaystyle\frac{2 |z_n|^{-2}}{1 - \sqrt{1 - 4|z_n|^{-2}}}, & |z_n| > \lambda, \\[2ex]
        0, & |z_n| \le \lambda,
    \end{cases}
\end{equation}
where
\begin{equation}\label{eq:zn_def}
    z_n \triangleq \frac{y_n}{\sigma}
    = \frac{x_n}{\sigma} + \xi_n
    \triangleq \eta_n + \xi_n, \qquad n = 1,\ldots,N,
\end{equation}
and $\eta_n \triangleq x_n / \sigma$ is the (complex-valued) signal normalized
by the noise standard deviation. In particular, we have
\begin{equation}
    |\eta_n|^2 = \mathsf{SNR}(x_n,\sigma),
\end{equation}
thus, $|z_n|^2$ is an estimator of the SNR of the $n$-th coordinate.

One may view~\eqref{eq:GSWest} as applying a scalar shrinkage
factor $g_\lambda(\cdot)$ to the LS estimate:
\begin{equation}
    \widehat{x}_{n,\GSW}
    = \widehat{x}_{n,\LS}\cdot \, g_\lambda(|z_n|),
\end{equation}
where
\begin{equation}\label{eq:g_lambda_def}
    g_\lambda(r)
    \triangleq
    \begin{cases}
        \displaystyle\frac{2 r^{-2}}{1 - \sqrt{1 - 4 r^{-2}}}, & r > \lambda, \\[2ex]
        0, & r \le \lambda,
    \end{cases}
    \qquad r \ge 0.
\end{equation}

\begin{remark}[Relation to the original SW estimator] The original SW estimator~\cite{weiss2021self} corresponds to the special case $\lambda = 2$ in~\eqref{eq:GSWest}. This choice yields a parameter-free, robust estimator that was shown to approximate the ideal Wiener filter in the high-SNR regime for the deconvolution setting. In the general denoising problem, however, fixing $\lambda = 2$ may be suboptimal at low SNR, where stronger noise suppression can be desirable. GSW introduces $\lambda$ as a tunable parameter, allowing to trade off bias and variance across SNR regimes.
\end{remark}

The intuition behind~\eqref{eq:GSWest} is as follows. The GSW rule applies a nonlinear shrinkage factor that aggressively zeroes out low-SNR components (those for which $|z_n| \le \lambda$) while preserving and gracefully attenuating
components above the threshold. Compared to classical thresholding, GSW has a more nuanced, data-dependent shrinkage profile that preserves a Wiener-like structure for large $|z_n|$ and provides additional robustness through the tunable parameter $\lambda$. Indeed, when $|\eta_n|\to\infty$, we have $|z_n|\to\infty$ almost surely, hence
\begin{align}
    &\frac{2|z_n|^{-2}}{1-\sqrt{1-4|z_n|^{-2}}}=\frac{1}{1+|z_n|^{-2}}+O_p(|z_n|^{-4}) \\
    &\qquad \Longrightarrow \; \widehat{x}_{n,\GSW}\approx \widehat{x}_{n,\LS}\cdot\frac{\mathsf{SNR}(x_n,\sigma)}{1+\mathsf{SNR}(x_n,\sigma)},
\end{align}
    i.e., the oracle MMSE solution.\footnote{Oracle, and not realizable, since we recall that $\mathsf{SNR}(x_n,\sigma)$ is a function of $x_n$, which is unknown.}

\subsection{Asymptotic MSE Analysis}\label{subsec:gws_analysis}
We now analyze the elementwise MSE of $\widehat{x}_{n,\GSW}$ \eqref{eq:GSWest},
\begin{equation}
    \mathsf{MSE}_n(\eta_n,\lambda)
    \triangleq \Eset\left[\left|\widehat{x}_{n,\GSW} - x_n\right|^2\right],
\end{equation}
as a function of the normalized signal $\eta_n = x_n/\sigma$ and the threshold $\lambda$. Our focus is on asymptotic expressions in the high- and low-SNR regimes, i.e., $|\eta_n| \to \infty$ and $|\eta_n| \to 0$, respectively.

\subsubsection{High-SNR regime}
In the high-SNR regime, the GSW estimator behaves as a perturbed oracle MMSE solution. The following result formalizes this behavior.

\begin{proposition}[High-SNR performance]\label{thm:highSNR_GSW}
Let $z_n$ be as in~\eqref{eq:zn_def} and define
\begin{equation}\label{eq:defofpprob}
    p_n \triangleq \Pr\left(|z_n| > \lambda\right)
    = Q_1\left(\sqrt{2} |\eta_n|, \sqrt{2} \lambda \right),
\end{equation}
where $Q_1(\cdot,\cdot)$ is the Marcum $Q$-function~\cite{corazza2002new}. Then, for $|\eta_n| \gg \lambda$,
the MSE of the estimator~\eqref{eq:GSWest} admits the expansion
\begin{equation}\label{eq:highSNR_MSE}
    \mathsf{MSE}_n(\eta_n,\lambda)
    = (1 - p_n)\,|x_n|^2 + p_n\,\sigma^2
      + \mathcal{O}\left( \frac{\sigma^2}{|\eta_n|^2} \right).
\end{equation}
\end{proposition}

\begin{IEEEproof}[Proof sketch]
A detailed derivation is provided in Appendix~\ref{app:highSNR_GSW}. The proof proceeds by conditioning on the events $\{|z_n| \le \lambda\}$ and $\{|z_n| > \lambda\}$, exploiting the noncentral $\chi^2$ distribution of $|z_n|^2$, and performing a high-SNR expansion in powers of $|\eta_n|^{-1}$.
\end{IEEEproof}

The interpretation of Proposition~\ref{thm:highSNR_GSW} is that, at high SNR, the GSW estimate closely approximates the oracle MMSE solution. Indeed, $Q_1\left(\sqrt{2}|\eta_n|,\sqrt{2}\lambda\right) \to 1$ as $|\eta_n| \to \infty$,
so the leading term in~\eqref{eq:highSNR_MSE} tends to $\sigma^2$, which is the MSE of the LS estimator in this deterministic setting, and the asymptotic MSE of the oracle MMSE solution. The probability of erroneously shrinking the coefficient to zero vanishes exponentially fast with $|\eta_n|^2$, and the term $\mathcal{O}(\sigma^2/|\eta_n|^2)$ becomes negligible.

\subsubsection{Low SNR regime}
In the low SNR regime, the dominant challenge is noise suppression. The next result characterizes the residual noise variance induced by the GSW rule \eqref{eq:GSWest}.

\begin{proposition}[Low-SNR performance]\label{thm:lowSNR_GSW}
For entries with low SNR, i.e., $|\eta_n| \ll 1$, the MSE of the GSW estimator
satisfies
\begin{equation}\label{eq:lowSNR_MSE}
    \mathsf{MSE}_n(\eta_n,\lambda)
    = |x_n|^2 + \rho_{\GSW}(\lambda)\,\sigma^2
      + \mathcal{O}\left(|\eta_n|\right),
\end{equation}
where $p_n = \Pr(|z_n| > \lambda)$ as in \eqref{eq:defofpprob} and
\begin{equation}\label{eq:rho_GSW_complex}
    \rho_{\GSW}(\lambda)
    \triangleq \left( \frac{\lambda^2 - 1}{2} \right) e^{-\lambda^2}
    + \frac{1}{2} \int_{\lambda^2}^{\infty} \sqrt{t^2 - 4t}\, e^{-t}\, \mathrm{d}t.
\end{equation}
\end{proposition}

\begin{IEEEproof}[Proof sketch]
A detailed derivation is provided in Appendix~\ref{app:lowSNR_GSW}. The proof
relies on expanding the noncentral $\chi^2$ density of $|z_n|^2$ around
$\eta_n = 0$, evaluating the resulting integrals term by term, and isolating the
leading-order contribution to the residual variance as a function of $\lambda$.
\end{IEEEproof}

Proposition~\ref{thm:lowSNR_GSW} shows that, at low SNR, GSW effectively suppresses
noise by pushing most coefficients below the threshold while controlling the
residual variance through $\rho_{\GSW}(\lambda)$. For suitably chosen $\lambda$,
$\rho_{\GSW}(\lambda)$ can be significantly smaller than the unit variance of
the LS estimator, as $\rho_{\GSW}(\lambda)\xrightarrow[\lambda\to\infty]{}0$, leading to substantial gains over naive denoising strategies.

\begin{remark}[Real-valued case]
Propositions~\ref{thm:highSNR_GSW} and~\ref{thm:lowSNR_GSW} extend naturally to the
real-valued case, where $\ux \in \Rset^{N \times 1}$ and
$\uxi \sim \mathcal{N}(\uo,\I_N)$. In that setting,
\[
    p_n = Q(\lambda - \eta_n) + Q(\lambda + \eta_n),
\]
where $Q(x) \triangleq \int_x^\infty \varphi(t)\,\mathrm{d}t$ with
$\varphi(x) \triangleq \frac{1}{\sqrt{2\pi}} e^{-x^2/2}$, and
\begin{equation}\label{eq:rho_GSW_real}
    \rho_{\GSW}(\lambda)
    \triangleq \left(\lambda + \sqrt{\lambda^2 - 4}\right) \varphi(\lambda)
    + e^{-2} Q\left(\sqrt{\lambda^2 - 4}\right) - Q(\lambda).
\end{equation}
The qualitative conclusions regarding the behavior of the GSW estimator remain unchanged.
\end{remark}

The high- and low-SNR regimes characterizations above highlight the role of GSW as a flexible nonlinear shrinkage estimator that balances signal preservation and noise suppression through the tunable parameter $\lambda$. In Section~\ref{sec:simulationresults}, we compare GSW to classical established benchmarks in a simulation experiment of a generic sparse vector denoising scenario.

\vspace{-0.2cm}
\section{Simulation Results}\label{sec:simulationresults}

We illustrate the behavior of GSW in sparse vector denoising. The signal $\ux$ has $K = 10$ non-zero entries (each of unit magnitude) out of $N = 1000$ coordinates. For each noise level, the empirical MSE is averaged over $10^3$ independent realizations. Figure~\ref{fig:simulationresults} shows the MSE versus the inverse noise variance $1/\sigma^2$ (equivalently, in this case, $\mathsf{SNR}(x_n,\sigma)$).

We compare GSW~\eqref{eq:GSWest} to LS~\eqref{eq:lsest}, ST~\cite{donoho1995denoising}, JS~\cite{draper1979ridge}, the original SW estimator~\cite{weiss2021self}, and the MMSE oracle lower bound from~\cite[Eq.~(13)]{weiss2021self}. The threshold is set to $\lambda = 1.1 \sqrt{2 \log N} \approx 4.09$, a slightly inflated universal threshold~\cite[Eq.~(32)]{donoho1995denoising}.

Evidently, GSW achieves in this setting a favorable trade-off across SNR regimes: it closely tracks the MMSE at high SNR while providing robust denoising at low SNR. In particular:
\begin{itemize}
    \item ST performs best among the classical thresholding rules at very low SNR, but GSW overtakes it around $\sim$2\,dB and maintains a noticeable advantage thereafter.
    \item GSW improves upon SW and JS across the entire SNR range, demonstrating the benefit of threshold tunability.
\end{itemize}

\begin{figure}
    \centering
    \includegraphics[width=\columnwidth]{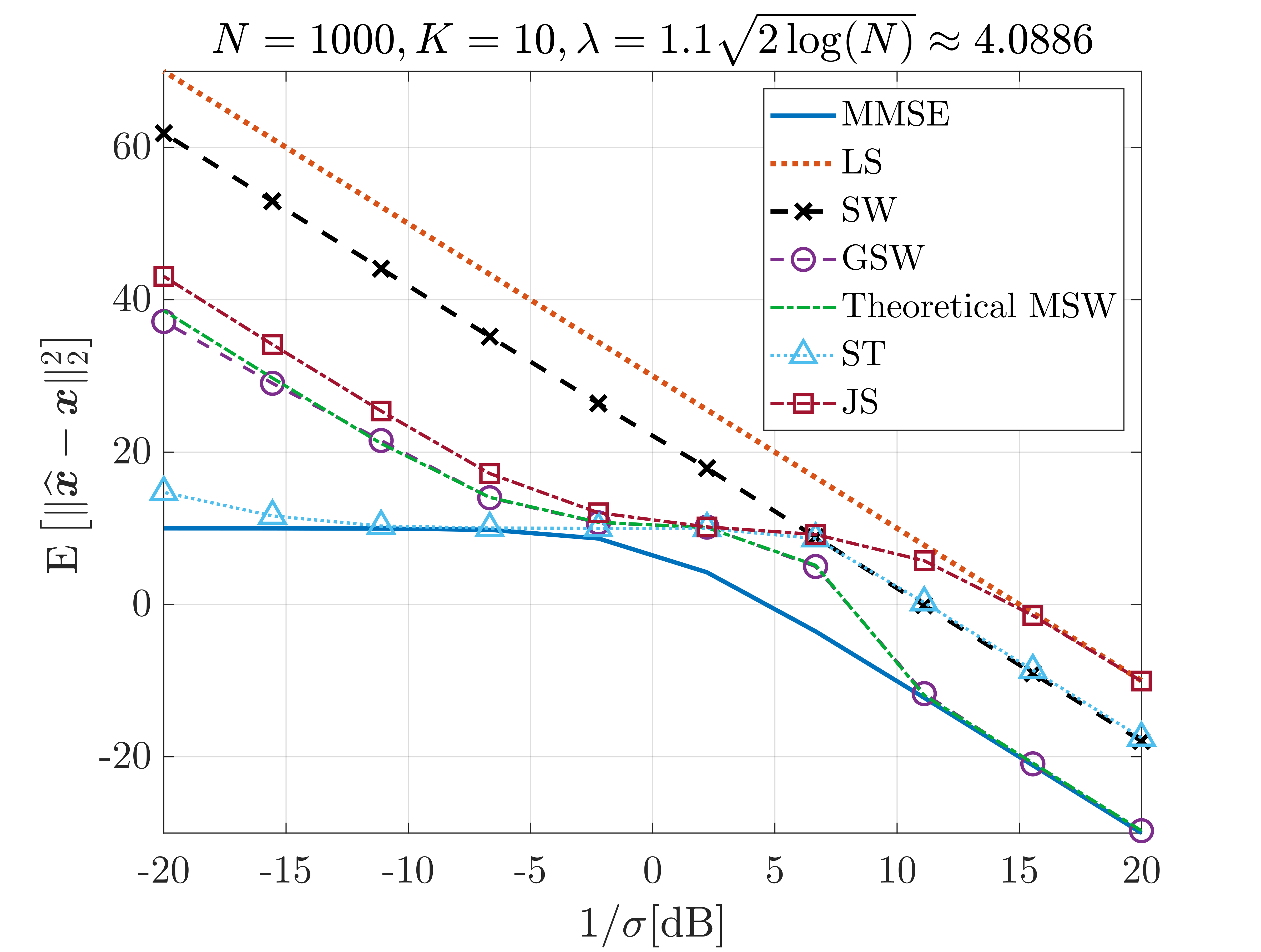}
    \caption{MSE versus inverse noise RMS power for a sparse signal with $K = 10$ non-zero entries (unit magnitude) out of $N = 1000$. Results are averaged over $10^3$ realizations. GSW approaches the oracle MMSE bound at high SNR and outperforms LS, JS, ST, and SW at medium-to-high SNR.}
    \label{fig:simulationresults}\vspace{-0.4cm}
\end{figure}

\section{Conclusion}\label{sec:conclusion}
We have introduced GSW filtering as a tunable-threshold-type shrinkage rule for AWGN denoising, extending the original SW estimator from deconvolution to the canonical signal-plus-noise model. GSW preserves the Wiener-like structure at high SNRs while enjoying an additional degree of freedom in the form of a threshold parameter that enables a controlled trade-off between bias and variance across SNR regimes. We derived closed-form high- and low-SNR MSE expressions, showing that GSW approaches the oracle MMSE performance at high SNR and achieves effective noise suppression at low SNR through the parameter-dependent constant $\rho_{\GSW}(\lambda)$. Simulations on a generic sparse denoising problem confirm these properties and demonstrate improvements over LS, JS, ST, and the original SW rule. These results suggest that GSW is a useful building block for practical transform-domain denoising, including sparse and other structured signal settings.

\appendices
\section{Proof of Proposition~\ref{thm:highSNR_GSW}}\label{app:highSNR_GSW}

We consider a single coordinate and omit the index $n$. Recall that
\begin{equation}
    y = x + \sigma \xi, \qquad
    z = \frac{y}{\sigma} = \eta + \xi,
\end{equation}
with $\xi \sim \mathcal{CN}(0,1)$ and $\eta \triangleq x/\sigma$. The scalar GSW estimate is
\begin{equation}\label{eq:GSW_scalar_app}
    \widehat{x}_{\GSW}
    = y \, g_\lambda(|z|)
    = \sigma(\eta + \xi) g_\lambda(|z|),
\end{equation}
where $g_\lambda(\cdot)$ is given by~\eqref{eq:g_lambda_def}. Define
\begin{equation}
    \mathcal{A} \triangleq \{|z| \le \lambda\},
    \qquad
    \mathcal{A}^{\mathrm{c}} \triangleq \{|z| > \lambda\},
\end{equation}
and
\begin{equation}
    p \triangleq \Pr\left(\mathcal{A}^{\mathrm{c}}\right)
      = \Pr(|z| > \lambda)
      = Q_1\left(\sqrt{2}|\eta|,\sqrt{2}\lambda\right).
\end{equation}

Using the indicator function $\mathbf{1}_{\{\cdot\}}$, the MSE decomposes as
\begin{align}
    \Eset\left[|\widehat{x}_{\GSW} - x|^2\right] &= \Eset\left[|\widehat{x}_{\GSW} - x|^2 \,\mathbf{1}_{\mathcal{A}}\right] \\
     &+ \Eset\left[|\widehat{x}_{\GSW} - x|^2 \,\mathbf{1}_{\mathcal{A}^{\mathrm{c}}}\right].
     \label{eq:MSE_split}
\end{align}

Given $\mathcal{A}$, we have $\widehat{x}_{\GSW} = 0$, hence
\begin{equation}\label{eq:MSE_A}
    \Eset\left[|\widehat{x}_{\GSW} - x|^2 \,\mathbf{1}_{\mathcal{A}}\right]
    = |x|^2\,\Pr(\mathcal{A})
    = (1 - p)\,|x|^2.
\end{equation}
Given $\mathcal{A}^{\mathrm{c}}$, $\widehat{x}_{\GSW}$ is given by~\eqref{eq:GSW_scalar_app} with $|z| > \lambda$, and we can write the error as
\begin{align}
    \widehat{x}_{\GSW} - x = \sigma z g_\lambda(|z|) - x. \label{eq:error_decomp}
\end{align}
Thus
\begin{align}\label{eq:MSE_Ac_def}
    \Eset\left[|\widehat{x}_{\GSW} - x|^2 \,\mathbf{1}_{\mathcal{A}^{\mathrm{c}}}\right] &= \sigma^2\Eset\left[| z g_\lambda(|z|)|^2\,\mathbf{1}_{\mathcal{A}^{\mathrm{c}}}\right] +p|x|^2 \\
    &-2\sigma \Re\left\{x^*\Eset\left[ z g_\lambda(|z|)\,\mathbf{1}_{\mathcal{A}^{\mathrm{c}}}\right]\right\}. \label{eq:MSE_Ac_def_corrected}
\end{align}

We now isolate the leading $1/|z|^2$ behavior of $g_\lambda(\cdot)$. Using
\begin{equation}
    g_\lambda(r)
    = \frac{2 r^{-2}}{1 - \sqrt{1 - 4 r^{-2}}}
    = \frac{1}{2}\left(1 + \sqrt{1 - 4 r^{-2}}\right), \quad r>2,
\end{equation}
a second-order Taylor expansion of $\sqrt{1-4r^{-2}}$ around $r^{-2}=0$ yields
\begin{equation}
    g_\lambda(r) \triangleq 1 - \frac{1}{r^2} + \delta_\lambda(r),
\end{equation}
with a remainder term satisfying
\begin{equation}\label{eq:delta_bound}
    \left|\delta_\lambda(r)\right| \le \frac{C_\lambda}{r^4},
    \qquad r \ge \lambda,
\end{equation}
for some finite constant $C_\lambda$ that depends on $\lambda$ but not on $\eta$ or $\sigma$. Substituting $g_\lambda(|z|) = 1 - |z|^{-2} + \delta_\lambda(|z|)$ into \eqref{eq:MSE_Ac_def}, we first expand the two expectations that appear there. On $\mathcal{A}^{\mathrm{c}}$,
\begin{align}
    z g_\lambda(|z|)
    &= z\left(1 - \frac{1}{|z|^2} + \delta_\lambda(|z|)\right) = z - \frac{z}{|z|^2} + z\,\delta_\lambda(|z|),
    \label{eq:zg_expand}
\end{align}
and
\begin{align}
    &|z g_\lambda(|z|)|^2 \\
    &= |z|^2\left(1 - \frac{1}{|z|^2} + \delta_\lambda(|z|)\right)
              \left(1 - \frac{1}{|z|^2} + \delta_\lambda(|z|)\right)^*\\
              &\triangleq |z|^2 - 2 + r_\lambda(z),
\end{align}
where, for some constant $K_\lambda<\infty$,
\begin{equation}\label{eq:r_bound}
    |r_\lambda(z)| \le \frac{K_\lambda}{|z|^2}, \qquad |z| \ge \lambda.
\end{equation}

Using \eqref{eq:zg_expand} we have
\begin{align}
    \Eset\left[z g_\lambda(|z|)\,\mathbf{1}_{\mathcal{A}^{\mathrm{c}}}\right]
    &= \Eset\left[z\,\mathbf{1}_{\mathcal{A}^{\mathrm{c}}}\right]
       - \Eset\left[\frac{z}{|z|^2}\,\mathbf{1}_{\mathcal{A}^{\mathrm{c}}}\right] \\
       &+ \Eset\left[z\,\delta_\lambda(|z|)\,\mathbf{1}_{\mathcal{A}^{\mathrm{c}}}\right],
    \label{eq:E_zg_split}
\end{align}
and from the expression for $|z g_\lambda(|z|)|^2$,
\begin{align}
    \Eset\left[|z g_\lambda(|z|)|^2\,\mathbf{1}_{\mathcal{A}^{\mathrm{c}}}\right]
    &= \Eset\left[|z|^2\mathbf{1}_{\mathcal{A}^{\mathrm{c}}}\right]- 2p
       + \Eset\left[r_\lambda(z)\,\mathbf{1}_{\mathcal{A}^{\mathrm{c}}}\right].
       \label{eq:E_zg2_split}
\end{align}

We now use asymptotic expansions of the various terms as $|\eta|\to\infty$.
Recall $z = \eta + \xi$ with $\xi\sim\mathcal{CN}(0,1)$. Then
\begin{equation}
    \Eset[z] = \eta, \qquad \Eset[|z|^2] = |\eta|^2 + 1.
\end{equation}
Moreover, the event $\{|z|\le\lambda\}$ lies in a fixed ball around the origin,
while $z$ is centered at $\eta$, hence there exist constants $c,C>0$ such that
\begin{equation}\label{eq:Ac_prob_small}
    \Pr(\mathcal{A}) = \Pr(|z|\le\lambda) \le C e^{-c|\eta|^2},
    \qquad |\eta|\to\infty.
\end{equation}
Since $\Eset\left[f(z)\,\mathbf{1}_{\mathcal{A}^{\mathrm{c}}}\right] = \Eset[f(z)] - \Eset\left[f(z)\,\mathbf{1}_{\mathcal{A}}\right]$, it follows that
\begin{align}
    \Eset\left[z\,\mathbf{1}_{\mathcal{A}^{\mathrm{c}}}\right]
    &= \eta + \mathcal{O}\left(e^{-c|\eta|^2}\right), \label{eq:Ez_Ac}\\
    \Eset\left[|z|^2\mathbf{1}_{\mathcal{A}^{\mathrm{c}}}\right]
    &= |\eta|^2 + 1 + \mathcal{O}\left(e^{-c|\eta|^2}\right), \label{eq:Ez2_Ac}\\
    p &= 1 + \mathcal{O}\left(e^{-c|\eta|^2}\right). \label{eq:p_Ac}
\end{align}

Next, we bound the terms involving $|z|^{-2}$ and $\delta_\lambda(\cdot)$. For $|z|\ge\lambda$, from \eqref{eq:delta_bound}, we have
\begin{equation}
    \left|z\,\delta_\lambda(|z|)\right|
    \le \frac{C_\lambda}{|z|^3}, \qquad
    |r_\lambda(z)| \le \frac{K_\lambda}{|z|^2}.
\end{equation}
Similarly to the technique in \cite[Lemma~1]{weiss2021self}, using Taylor expansions around $z=\eta$, one obtains
\begin{align}
    \Eset\left[\frac{z}{|z|^2}\,\mathbf{1}_{\mathcal{A}^{\mathrm{c}}}\right]
    &= \frac{1}{\eta^*}
       + \mathcal{O}\left(\frac{1}{|\eta|^3}\right), \label{eq:Ez_over_r2}\\
    \Eset\left[z\,\delta_\lambda(|z|)\,\mathbf{1}_{\mathcal{A}^{\mathrm{c}}}\right]
    &= \mathcal{O}\left(\frac{1}{|\eta|^3}\right), \label{eq:Ez_delta}\\
    \Eset\left[r_\lambda(z)\,\mathbf{1}_{\mathcal{A}^{\mathrm{c}}}\right]
    &= \mathcal{O}\left(\frac{1}{|\eta|^2}\right), \label{eq:Er_lambda}
\end{align}
as $|\eta|\to\infty$ by using that $\xi$ has finite moments of all orders and that $\Pr(\mathcal{A})$ (in \eqref{eq:Ac_prob_small}) is exponentially small in $|\eta|$.

Substituting \eqref{eq:Ez_Ac}, \eqref{eq:Ez_over_r2}, and \eqref{eq:Ez_delta}
into \eqref{eq:E_zg_split}, and recalling $x = \sigma\eta$, we get
\begin{align}
    \Eset\left[z g_\lambda(|z|)\,\mathbf{1}_{\mathcal{A}^{\mathrm{c}}}\right]
    &= \eta
       - \frac{1}{\eta^*}
       + \mathcal{O}\left(\frac{1}{|\eta|^3}\right),
\end{align}
and hence
\begin{align}
    &-2\sigma\,\Re\left\{
         x^* \Eset\left[z g_\lambda(|z|)\,\mathbf{1}_{\mathcal{A}^{\mathrm{c}}}\right]
       \right\} \\
    &= -2\sigma\,\Re\left\{
         (\sigma\eta)^*
         \left(\eta - \frac{1}{\eta^*}\right)
       \right\}
       + \mathcal{O}\left(\frac{\sigma^2}{|\eta|^2}\right) \\
    &= -2|x|^2 +2\sigma^2 + \mathcal{O}\left(\frac{\sigma^2}{|\eta|^2}\right).
    \label{eq:first_term_final}
\end{align}

Similarly, substituting \eqref{eq:Ez2_Ac}, \eqref{eq:p_Ac}, and
\eqref{eq:Er_lambda} into \eqref{eq:E_zg2_split} yields
\begin{align}
    \Eset\left[|z g_\lambda(|z|)|^2\,\mathbf{1}_{\mathcal{A}^{\mathrm{c}}}\right]
    &= |\eta|^2 - 1 + \mathcal{O}\left(\frac{1}{|\eta|^2}\right),
\end{align}
and therefore
\begin{equation}\label{eq:second_term_final}
    \sigma^2 \Eset\left[|z g_\lambda(|z|)|^2\,\mathbf{1}_{\mathcal{A}^{\mathrm{c}}}\right]
    = |x|^2 - \sigma^2 + \mathcal{O}\left(\frac{\sigma^2}{|\eta|^2}\right).
\end{equation}

Combining \eqref{eq:first_term_final}, \eqref{eq:second_term_final} and
\begin{align}
    p|x|^2 &= |x|^2 + \mathcal{O}\left(\frac{\sigma^2}{|\eta|^2}\right),\\
    \sigma^2 &= p\sigma^2 +(1-p)\sigma^2 = p\sigma^2 + \mathcal{O}\left(e^{-c|\eta|^2}\right),
\end{align}
in \eqref{eq:MSE_Ac_def_corrected}, we obtain
\begin{equation}\label{eq:MSE_Ac_final}
    \Eset\left[|\widehat{x}_{\GSW} - x|^2 \,\mathbf{1}_{\mathcal{A}^{\mathrm{c}}}\right]
    = p\,\sigma^2 + \mathcal{O}\left(\frac{\sigma^2}{|\eta|^2}\right).
\end{equation}

Substituting \eqref{eq:MSE_A} and \eqref{eq:MSE_Ac_final} into the decomposition
\eqref{eq:MSE_split} yields
\begin{equation}
    \Eset\left[|\widehat{x}_{\GSW} - x|^2\right]
    = (1 - p)\,|x|^2 + p\,\sigma^2
      + \mathcal{O}\left(\frac{\sigma^2}{|\eta|^2}\right),
\end{equation}
which is precisely \eqref{eq:highSNR_MSE} in Proposition~\ref{thm:highSNR_GSW}.

\section{Proof of Proposition~\ref{thm:lowSNR_GSW}}\label{app:lowSNR_GSW}

We again work with a single coordinate, omit the index $n$ and use the same definitions as in Appendix \ref{app:highSNR_GSW}. Recall the MSE decomposition 
\begin{align}
    \hspace{-0.075cm}\Eset\left[|\widehat{x}_{\GSW} - x|^2\right]
    &= (1-p)\,|x|^2
      + p\,\Eset\left[|\widehat{x}_{\GSW} - x|^2 \mid \mathcal{A}^{\mathrm{c}}\right].
    \label{eq:MSE_low_split}
\end{align}
Since $z\sim\mathcal{CN}(\eta,1)$, the probability $p = \Pr(|z|>\lambda)$
is analytic in $|\eta|$ and, expanding around $\eta=0$, we obtain
(cf. \cite[Eq.~(50)--(51)]{weiss2021self})
\begin{equation}
    p = p_0 + \mathcal{O}\left(|\eta|^2\right), \qquad
    p_0 \triangleq \Pr(|\xi|>\lambda) \in (0,1).
    \label{eq:p_low_expansion}
\end{equation}

Write the conditional MSE conditioned on $\mathcal{A}^{\mathrm{c}}$ as
\begin{align}
    \Eset\left[|\widehat{x}_{\GSW} - x|^2 \mid \mathcal{A}^{\mathrm{c}}\right]
    &= |x|^2
       + \sigma^2\Eset\left[| z g_\lambda(|z|)|^2 \mid \mathcal{A}^{\mathrm{c}}\right] \label{eq:abssquaredterm}\\
       &- 2\sigma^2\Re\left\{\eta^* \Eset\left[z g_\lambda(|z|) \mid \mathcal{A}^{\mathrm{c}}\right]\right\}.
       \label{eq:MSE_low_cond_expand}
\end{align}
Since the map $\eta \mapsto \Eset[z g_\lambda(|z|) \mid \mathcal{A}^{\mathrm{c}}]$ is analytic, we may write it (near $\eta=0$) as
\begin{equation}
    \Eset\left[z g_\lambda(|z|)  \mid \mathcal{A}^{\mathrm{c}}\right] = \left.\Eset\left[z g_\lambda(|z|)  \mid \mathcal{A}^{\mathrm{c}}\right]\right|_{\eta=0} + \mathcal{O}\left(|\eta|\right).
\end{equation}
However, at $\eta=0$ the distribution of $z$ is rotationally symmetric ($\left.z\right|_{\eta=0}=\xi\sim\mathcal{CN}(0,1)$) while $g_\lambda(|z|)$ depends only on $|z|$. Hence
\begin{equation}
    \left.\Eset\left[z g_\lambda(|z|)  \mid \mathcal{A}^{\mathrm{c}}\right]\right|_{\eta=0} = 0 \, \Rightarrow \, \Eset\left[z g_\lambda(|z|)  \mid \mathcal{A}^{\mathrm{c}}\right] = \mathcal{O}\left(|\eta|\right),
    \label{eq:Exhat_low}
\end{equation}
so the cross term in~\eqref{eq:MSE_low_cond_expand} contributes $\mathcal{O}(|x|\,|\eta|)=\mathcal{O}(|\eta|^2)$.

Proceeding to the second term in \eqref{eq:abssquaredterm}, we have
\begin{align}
    &\sigma^2\Eset\left[| z g_\lambda(|z|)|^2 \mid \mathcal{A}^{\mathrm{c}}\right]\\
    &= \sigma^2\,
       \Eset\left[|\xi|^2 |g_\lambda(|\xi|)|^2 \mid |\xi|>\lambda\right]
       + \mathcal{O}\left(|\eta|\right) \\
    &= \sigma^2\cdot\frac{\rho_{\GSW}(\lambda)}{p_0} + \mathcal{O}\left(|\eta|\right),
    \label{eq:Exhat2_low}
\end{align}
where, by explicit radial integration (see~\cite[Appendix~A]{weiss2021self}), we obtain the constant \eqref{eq:rho_GSW_complex},
which depends only on $\lambda$.

Substituting~\eqref{eq:Exhat_low} and \eqref{eq:Exhat2_low} into \eqref{eq:MSE_low_cond_expand} gives
\begin{equation}
    \Eset\left[|\widehat{x}_{\GSW} - x|^2 \mid \mathcal{A}^{\mathrm{c}}\right]
    = |x|^2 + \rho_{\GSW}(\lambda)\sigma^2
      + \mathcal{O}\left(|\eta|\right).
\end{equation}
Finally, inserting this and~\eqref{eq:p_low_expansion} into
\eqref{eq:MSE_low_split} yields
\begin{align}
    \Eset\left[|\widehat{x}_{\GSW} - x|^2\right]
    &= (1-p)\,|x|^2\\
       &+ p\left(|x|^2 + \sigma^2\cdot\frac{\rho_{\GSW}(\lambda)}{p_0}
                 + \mathcal{O}\left(|\eta|\right)\right) \\
    &= |x|^2 + \sigma^2\cdot\rho_{\GSW}(\lambda)
       + \mathcal{O}\left(|\eta|\right),
\end{align}
which is the claimed low-SNR expansion \eqref{eq:lowSNR_MSE}.

\bibliographystyle{IEEEbib}
\bibliography{./Inputs/refs}

\end{document}